\documentclass[11pt,expanded]{article}
\usepackage{verbatim}
\usepackage{color}
\usepackage{amssymb,amsmath,amsthm}
\usepackage{geometry}
\usepackage{graphicx}
\usepackage{subfigure}

\usepackage{times}
\usepackage{epsfig}

\usepackage{indentfirst}
\usepackage{cite}
\usepackage{amsfonts}
\usepackage{array}
\usepackage{multirow}

\usepackage{xcolor}
\usepackage{authblk}

\title{Title}

\usepackage{extarrows}
\usepackage{url}
\usepackage{mathrsfs}
\usepackage{caption}
\usepackage{setspace}
\usepackage{booktabs}
\usepackage[round,authoryear]{natbib}
\usepackage{amsthm}
\theoremstyle{definition}
\newtheorem{definition}{Definition}
\newtheorem{lemma}{Lemma}
\newtheorem{theorem}{Theorem}
\newtheorem{corollary}{Corollary}
\newtheorem{remark}{Remark}

\usepackage{hyperref}

\usepackage[normalem]{ulem}
\usepackage{titlesec} 

\titleformat{\section}
  {\Large\bfseries} 
  {\thesection}{1em}{} 

\usepackage[ruled,vlined]{algorithm2e}

\newcommand{\ie}{{\it i.e.}}
\newcommand{\eg}{{\it e.g.}}

\begin{document}
\title{General sample size analysis for probabilities of causation: a delta method approach}
\date{}
\author[1]{Tianyuan Cheng}
\author[2]{Ruirui Mao}
\author[3]{Judea Pearl}
\author[4]{Ang Li}

\affil[1]{Department of Statistics, Florida State University}
\affil[2]{Department of Computer Science, Cornell University}
\affil[3]{Department of Computer Science, University of California Los Angeles}
\affil[4]{Department of Computer Science, Florida State University}
\maketitle

\begin{abstract}
Probabilities of causation (PoCs), such as the probability of necessity and sufficiency (PNS), are important tools for decision making but are generally not point identifiable. Existing work has derived bounds for these quantities using combinations of experimental and observational data. However, there is very limited research on sample size analysis, namely, how many experimental and observational samples are required to achieve a desired margin of error. In this paper, we propose a general sample size framework based on the delta method. Our approach applies to settings in which the target bounds of PoCs can be expressed as finite minima or maxima of linear combinations of experimental and observational probabilities. Through simulation studies, we demonstrate that the proposed sample size calculations lead to stable estimation of these bounds.
\end{abstract}

\section{Introduction}

Probabilities of causation (PoCs) are used in many real-world applications, such as marketing, law, social science, and health science, especially when decisions depend on whether an action caused an outcome. For example, \cite{li2022unit} introduced a ``benefit function” that is a linear combination of PoCs and reflects the payoff or cost of selecting an individual with certain features, with the goal of finding those most likely to show a target behavior. \cite{stott2004human} used it in climate event assignment to quantify how much human influence changes the risk of an extreme event. \cite{mueller2023personalized} argued that PoCs can be used for personalized decision making, and \cite{li2020training} found that PoCs can be helpful for improving the accuracy of some machine learning methods.

Without extra assumptions, PoCs are generally not identifiable, so one often works with bounds instead of point values. Using the structural causal model (SCM), \cite{pearl1999probabilities} defined three binary PoCs, including PNS, PN, and PS. \cite{tian2000probabilities} derived bounds for these quantities using both experimental and observational information, and later \cite{li2019unit,li2024probabilities} provided formal proofs. Several papers studied how to tighten these bounds. For example, \cite{mueller2021causes} used covariates and causal structure to narrow the bounds for PNS, and \cite{dawid2017probability} used covariates to narrow the bounds for PN. 

Most of these works implicitly assume that the experimental and observational samples are large enough to estimate the needed probabilities well. However, there is little work on how to choose sample sizes so that the estimated bounds have a desired level of precision. This gap limits the use of the theory in real applications. \cite{li2022probabilities} discussed this issue, but focused on a special case for the PNS bound. To our knowledge, there is still no unified framework that links a target error level to required experimental and observational sample sizes for general PoC bounds. 

In this paper, we study adequate sample sizes for estimating bounds of PoCs from a general perspective. Our starting point is that many sharp bounds can be written as a finite minimum or maximum of explicit functions of a finite set of probabilities, including experimental probabilities such as $P(y_x)$ and observational probabilities such as $P(x,y)$. For PNS, the bound components are typically linear in these probabilities; for PN and PS, the bound components often take a ratio form. Under mild regularity conditions (e.g., denominators bounded away from zero), the bound components are smooth transformations of the underlying probability vector. The resulting lower and upper bounds can still be non-smooth because they are formed by finite minima or maxima. The fact that the components are smooth but the bounds can be non-smooth leads to different asymptotic behavior in the two cases, which our results accommodate.

Our main contributions are:
\begin{itemize}
\item We propose a general sample size framework for estimating PoC bound endpoints (i.e., lower and upper bounds) with a pre-specified margin of error, based on multivariate delta-method variance approximations for smooth endpoints, and a directional delta method implemented via numerical methods for non-smooth endpoints, following \cite{fang2019inference}.
\item Our asymptotic results apply whenever the bound endpoints can be expressed as finite minima or maxima of explicit bound components. This covers common forms in the PoC literature and also extends to other bounded causal quantities, including linear combinations of PoCs.
\item We provide simulation studies showing that the proposed sample sizes are stable and sufficient in practice, and that they are far less conservative than existing results in the literature.
\end{itemize}

\section{Preliminaries}\label{sec2}
In this section, we briefly review the definitions of the three aspects of binary causation following \cite{tian2000probabilities}. Our analysis is based on the counterfactual framework within structural causal models (SCMs) as introduced in \cite{pearl2009causality}. We denote by $Y_x = y$ the counterfactual statement that variable $Y$ would take value $y$ if $X$ were set to $x$. Throughout the paper, we use $y_x$ to represent the event $Y_x = y$, $y_{x'}$ for $Y_{x'} = y$, $y_x'$ for $Y_x = y'$, and $y_{x'}'$ for $Y_{x'} = y'$. Experimental information is summarized through causal quantities such as $P(y_x)$, while observational information is summarized by joint distributions such as $P(x, y)$. Unless otherwise stated, $X$ denotes the treatment variable and $Y$ denotes the outcome variable.

For the following three probabilities of causation, we all assume that $X$ and $Y$ are two binary variables in a causal model $M$. Let $x$ and $y$ stand for the propositions $X=\text{true}$ and $Y=\text{true}$, respectively, and $x'$ and $y'$ for their complements. Then we have:

\begin{definition}[Probability of Necessity (PN)]\label{defPN}
The probability of necessity is defined as:
\begin{equation}
\begin{aligned}
    \text{PN}
\;\triangleq\; &
P(Y_{x'} = \text{false} \mid X=\text{true}, Y=\text{true}) \\
\;\triangleq\;&
P(y_{x'}' \mid x, y).
\end{aligned}
\end{equation}
\end{definition}

\begin{definition}[Probability of Sufficiency (PS)]\label{defPS}
The probability of sufficiency is defined as:
\begin{equation}
\begin{aligned}
    \text{PS}
\;\triangleq\;&
P(Y_x = \text{true} \mid X=\text{false}, Y=\text{false})\\
\;\triangleq\;&
P(y_x \mid x', y').
\end{aligned}    
\end{equation}
\end{definition}

\begin{definition}[Probability of Necessity and Sufficiency (PNS)]\label{defPNS}
The probability of necessity and sufficiency is defined as:
\begin{equation}
\begin{aligned}
    \text{PNS}
\;\triangleq\;&
P(Y_x = \text{true},\, Y_{x'} = \text{false})\\
\;\triangleq\;&
P(y_x, y_{x'}').
\end{aligned}
\end{equation}
\end{definition}

We also introduce the sharp bounds on the PoCs derived from experimental and observational data in \cite{tian2000probabilities}, where the bounds are obtained via Balke’s program in \cite{balke1995probabilistic}. The bounds are:

\begin{equation}\label{PNSbound}
 \max\left\{
\begin{array}{l}
0,\\
P(y_x) - P(y_{x'}),\\
P(y) - P(y_{x'}),\\
P(y_x) - P(y)
\end{array}
\right\} \le \mathrm{PNS} \;\le\;
\min\left\{
\begin{array}{l}
P(y_x),\\
P(y_{x'}),\\
P(x,y) + P(x',y'),\\
P(y_x) - P(y_{x'}) + P(x,y') + P(x',y)
\end{array}
\right\}.
\end{equation}

\begin{equation}\label{PNbound}
\max\left\{
0,\;
\frac{P(y)-P(y_{x'})}{P(x,y)}
\right\}
\le \mathrm{PN}
\le
\min\left\{
1,\;
\frac{P(y'_{x'})-P(x',y')}{P(x,y)}
\right\}.
\end{equation}

\begin{equation}\label{PSbound}
\max\left\{
0,\;
\frac{P(y_x)-P(y)}{P(x',y')}
\right\}
\le \mathrm{PS}
\le
\min\left\{
1,\;
\frac{P(y_x)-P(x,y)}{P(x',y')}
\right\}.
\end{equation}

\section{Main Results}

Let $\theta$ be a parameter vector that collects all experimental and observational probabilities we need. For example, in the setting of the upper bound of \eqref{PNSbound} we can take 
\begin{equation}
    \theta=( P(y_x),P(y_{x'}),P(x,y), P(x',y'), P(x,y'), P(x',y))^\top.
\end{equation}

We first claim that in most scenarios, the bounds of PoCs can be written in the form of a complex linear combination of $\theta$.

\begin{lemma}[Piecewise-linear and fractional forms of sharp bounds for PoCs]\label{lem1}
Consider a structural causal model with finite discrete variables, and focus on the case that $X$ and $Y$ are both binary, \ie, $X$ has two treatment levels $x, x'$ and $Y$ has two outcome levels $y, y'$. Let $\theta\in\mathbb{R}^d$ collect the observational and experimental probabilities used to constrain the model.

For a probability of causation $Q\in\{\mathrm{PNS},\mathrm{PN},\mathrm{PS}\}$ as defined in Section $2$, let $U_{\mathrm{Q}}(\theta)$ and $L_{\mathrm{Q}}(\theta)$ be its sharp upper and lower bounds over all SCMs. Then there exist finite positive integers $J,K<\infty$, vectors $a_j,c_k\in\mathbb{R}^d$ and constants $b_j,d_k\in\mathbb{R}$ such that
\begin{equation}\label{equ:affineform}
    U_{\mathrm{Q}}(\theta)=\frac{1}{h_Q(\theta)}\min_{1\le j\le J}\{a_j^\top\theta+b_j\},\,\, \,L_{\mathrm{Q}}(\theta)=\frac{1}{h_Q(\theta)}\max_{1\le k\le K}\{c_k^\top\theta+d_k\},
\end{equation}
where the denominator $h_Q(\theta)$ is defined as:
\begin{enumerate}
    \item For PNS: $h_{\text{PNS}}(\theta) = 1$ (reducing to a piecewise-linear form).
    \item For PN: $h_{\text{PN}}(\theta) = P(x, y)$, which is a specific component of $\theta$.
    \item For PS: $h_{\text{PS}}(\theta) = P(x', y')$, which is a specific component of $\theta$.
\end{enumerate}
$U_Q(\theta)$ and $L_Q(\theta)$ are continuous and piecewise-defined functions of $\theta$. They are differentiable at any $\theta$ where the optimizer is unique.
\end{lemma}
\begin{proof}
According to \cite{tian2000probabilities}, for finite discrete SCMs, sharp bounds for PNS (and the numerators of PN, PS) can be obtained by optimizing a linear functional over the set of latent-type distributions consistent with $\theta$, which is a linear programming (LP) problem. Then, by LP duality and the polyhedral structure of the dual feasible set (see \cite{boyd2004convex}), the optimal value is a piecewise-linear function of $\theta$ and has a finite min or max representation of linear functions $a^\top \theta + b$. 
\end{proof}

Lemma $1$ gives a unified form for the sharp bounds: for PNS they are piecewise-linear in $\theta$ while for PN and PS they are piecewise linear-fractional due to normalization by $P(x, y)$ and $P(x', y')$. In either case, the bound endpoints are piecewise smooth and are differentiable whenever the active term is unique (optimizer is unique), and thus the inference for the plug-in estimators $\widehat{U_Q}=U_Q(\widehat{\theta})$ and $\widehat{L_Q}=L_Q(\widehat{\theta})$ reduces to studying piecewise-defined functions of
$\widehat{\theta}$. Next, we establish asymptotic normality and confidence intervals in the regular case where the optimizer is unique.

We first state the setup: let $m$ be the experimental sample size and $n$ be the observational sample size.
Assume $r$ is the ratio of experimental sample size to observational sample size, \ie, $m/n\to r\in(0,\infty)$. Let $\theta=(\theta_e^\top,\theta_o^\top)^\top\in\mathbb{R}^d$ collect the experimental and observational probabilities used in Lemma $1$, with true value $\theta_0=(\theta_{e0}^\top,\theta_{o0}^\top)^\top\in\mathbb{R}^d$. Let $\hat\theta=(\hat\theta_e^\top,\hat\theta_o^\top)^\top$ be the plug-in estimator based on
the two independent samples. By the multivariate central limit theorem (CLT) for multinomial proportions,
\begin{equation}
\sqrt{m}\,(\hat\theta_e-\theta_{e0})\ \Rightarrow\ N(0,\Sigma_e), \qquad
\sqrt{n}\,(\hat\theta_o-\theta_{o0})\ \Rightarrow\ N(0,\Sigma_o),
\end{equation}
and the two limits are independent. Equivalently,
\begin{equation}
\sqrt{n}\,(\hat\theta-\theta_0)\ \Rightarrow\ N(0,\Omega_r),
\,\,
\Omega_r=\begin{pmatrix}\Sigma_e/r & 0\\ 0 & \Sigma_o\end{pmatrix}.
\end{equation}

\begin{theorem}\label{thm1}
Assume $h_Q(\theta_0)>0$. Let 
\begin{equation*}
j^*=\arg\min_{1\le j\le J}\{a_j^\top\theta_0+b_j\}, \qquad k^*=\arg\max_{1\le k\le K}\{c_k^\top\theta_0+d_k\},
\end{equation*}
and assume both optimizers are unique with positive gaps:
\begin{equation*}
\Delta_U:=\min_{j\neq j^*}(a_j^\top\theta_0+b_j)-(a_{j^*}^\top\theta_0+b_{j^*})>0,
\end{equation*}
\begin{equation*}
\Delta_L:=(c_{k^*}^\top\theta_0+d_{k^*})-\max_{k\neq k^*}(c_k^\top\theta_0+d_k)>0.
\end{equation*}

Suppose $\sqrt{n}\,(\hat\theta-\theta_0)\ \Rightarrow\ N(0,\Omega_r)$. Then $U_Q$ and $L_Q$ are differentiable at $\theta_0$ and
\begin{equation}
\begin{aligned}
&\sqrt{n}\,(\hat U_Q-U_Q(\theta_0))\ \Rightarrow\
N\!\left(0,\ \nabla U_Q(\theta_0)^\top\Omega_r\,\nabla U_Q(\theta_0)\right),\\
&\sqrt{n}\,(\hat L_Q-L_Q(\theta_0))\ \Rightarrow\
N\!\left(0,\ \nabla L_Q(\theta_0)^\top\Omega_r\,\nabla L_Q(\theta_0)\right).
\end{aligned}
\end{equation}
\qed
\end{theorem}

\begin{proof}
We prove the result for the upper bound $U_Q$ only and the proof for lower bound is similar.

Define
\begin{equation*}
g(\theta):=\min_{1\le j\le J}\{a_j^\top\theta+b_j\},
\qquad
U_Q(\theta)=\frac{g(\theta)}{h_Q(\theta)}.    
\end{equation*}
By assumption, the minimizer at $\theta_0$ is unique and has a positive gap:
\begin{equation*}
\Delta_U=\min_{j\neq j^*}\bigl(a_j^\top\theta_0+b_j\bigr)-\bigl(a_{j^*}^\top\theta_0+b_{j^*}\bigr)>0.   
\end{equation*}
Since each $a_j^\top\theta+b_j$ is continuous and the index set is finite, the gap implies that there exist a $\epsilon >0$ such that for all $\|\theta-\theta_0\|\le \epsilon$, the minimizer is $j^*$. Hence $g(\theta)$ is differentiable at $\theta_0$. By Danskin's theorem (\cite{danskin2012theory}), 
\begin{equation*}
\nabla g(\theta_0)=a_{j^*}.    
\end{equation*}

Because $h_Q$ is differentiable at $\theta_0$ and $h_Q(\theta_0)>0$, it follows that $U_Q$ is differentiable at $\theta_0$ with
\begin{equation*}
\nabla U_Q(\theta_0)
=\frac{h_Q(\theta_0)\nabla g(\theta_0)-g(\theta_0)\nabla h_Q(\theta_0)}{h_Q(\theta_0)^2}
=\frac{h_Q(\theta_0)a_{j^*}-(a_{j^*}^\top\theta_0+b_{j^*})\nabla h_Q(\theta_0)}{h_Q(\theta_0)^2}.
\end{equation*}

Now apply the multivariate delta method to the differentiable map
$\theta\mapsto U_Q(\theta)$ at $\theta_0$:
\begin{equation}
    \sqrt n\bigl(U_Q(\hat\theta)-U_Q(\theta_0)\bigr)
\ \Rightarrow\
N\Bigl(0,\ \nabla U_Q(\theta_0)^\top\,\Omega_\tau\,\nabla U_Q(\theta_0)\Bigr).
\end{equation}

\end{proof}

Using Theorem \ref{thm1} we can derive a corresponding confidence interval as follows:
\begin{corollary}\label{coro1}
Assume the conditions of Theorem $1$. Let $\widehat{\Omega}_r$ be a consistent estimator of $\Omega_r$.
Define the plug-in variance estimators
\begin{equation*}
    \widehat{V}_U := \nabla U_Q(\widehat{\theta})^\top \widehat{\Omega}_r \nabla U_Q(\widehat{\theta}),
    \qquad
    \widehat{V}_L := \nabla L_Q(\widehat{\theta})^\top \widehat{\Omega}_r \nabla L_Q(\widehat{\theta}),
\end{equation*}
and the corresponding standard errors
\begin{equation*}
 \widehat{\mathrm{SE}}(\widehat{U}_Q) := \sqrt{\widehat{V}_U/n},
\qquad
\widehat{\mathrm{SE}}(\widehat{L}_Q) := \sqrt{\widehat{V}_L/n}.   
\end{equation*}
Then, asymptotic $(1-\alpha)$ confidence intervals for $U_Q(\theta_0)$ and $L_Q(\theta_0)$ are
\begin{equation}
CI_U=\left[\widehat{U}_Q - z_{1-\alpha/2}\,\widehat{\mathrm{SE}}(\widehat{U}_Q),\, \widehat{U}_Q + z_{1-\alpha/2}\,\widehat{\mathrm{SE}}(\widehat{U}_Q) \right],
\end{equation}
\begin{equation}
CI_L= \left[\widehat{L}_Q - z_{1-\alpha/2}\,\widehat{\mathrm{SE}}(\widehat{L}_Q), \, \widehat{L}_Q + z_{1-\alpha/2}\,\widehat{\mathrm{SE}}(\widehat{L}_Q) \right] ,    
\end{equation}
where $z_{1-\alpha/2}$ can be found on z-table of standard normal distribution.\qed
\end{corollary}

\begin{proof}
By Theorem 1, we have
\begin{equation*}
\sqrt{n}\bigl(\widehat{U}_Q - U_Q(\theta_0)\bigr) \Rightarrow N(0,V_U),
\qquad
V_U := \nabla U_Q(\theta_0)^\top \Omega_r \nabla U_Q(\theta_0),    
\end{equation*}
and similarly
\begin{equation*}
    \sqrt{n}\bigl(\widehat{L}_Q - L_Q(\theta_0)\bigr) \Rightarrow N(0,V_L),
\qquad
V_L := \nabla L_Q(\theta_0)^\top \Omega_r \nabla L_Q(\theta_0).
\end{equation*}

Let $\widehat{\Omega}_r$ be a consistent estimator of $\Omega_r$. Under the uniqueness conditions in
Theorem 1, $U_Q$ and $L_Q$ are differentiable in a neighborhood of $\theta_0$, hence
$\nabla U_Q(\cdot)$ and $\nabla L_Q(\cdot)$ are continuous at $\theta_0$. Since
$\widehat{\theta} \xrightarrow{p} \theta_0$, we have
\begin{equation*}
\nabla U_Q(\widehat{\theta}) \xrightarrow{p} \nabla U_Q(\theta_0),
\qquad
\nabla L_Q(\widehat{\theta}) \xrightarrow{p} \nabla L_Q(\theta_0).    
\end{equation*}

Therefore, by Slutsky's theorem (\cite{van2000asymptotic}),
\begin{equation*}
\widehat{V}_U := \nabla U_Q(\widehat{\theta})^\top \widehat{\Omega}_r \nabla U_Q(\widehat{\theta})
\xrightarrow{p} V_U,
\qquad
\widehat{V}_L := \nabla L_Q(\widehat{\theta})^\top \widehat{\Omega}_r \nabla L_Q(\widehat{\theta})
\xrightarrow{p} V_L.    
\end{equation*}

Define $\widehat{\mathrm{SE}}(\widehat{U}_Q) = \sqrt{\widehat{V}_U/n}$ and
$\widehat{\mathrm{SE}}(\widehat{L}_Q) = \sqrt{\widehat{V}_L/n}$. Then
\begin{equation}
    \frac{\widehat{U}_Q - U_Q(\theta_0)}{\widehat{\mathrm{SE}}(\widehat{U}_Q)} \Rightarrow N(0,1),
\qquad
\frac{\widehat{L}_Q - L_Q(\theta_0)}{\widehat{\mathrm{SE}}(\widehat{L}_Q)} \Rightarrow N(0,1),
\end{equation}
which gives an asymptotic $(1-\alpha)$ confidence intervals.
\end{proof}

When the unique-optimizer assumption in Theorem $1$ fails, means that when we have multiple active constraints, then the bound function may be non-differentiable at $\theta_0$. In this case we use a directional delta method (see \cite{dumbgen1993nondifferentiable}, \cite{fang2019inference}) to derive the limiting distribution. The following theorem states the result.

\begin{theorem}\label{thm2}
Let $U_Q(\theta)$ and $L_Q(\theta)$ be defined as in equation \eqref{equ:affineform}. Assume $h_Q(\theta_0)>0$ and $h_Q$ is differentiable at $\theta_0$. Define the active sets as follows:
\begin{equation*}
J_0= \arg \min_j ( \ a_j^\top\theta_0+b_j), \qquad K_0= \arg \max_k ( c_k^\top\theta_0+d_k).
\end{equation*}

If $\sqrt{n}\bigl(\hat{\theta}-\theta_0\bigr)\Rightarrow Z\sim N\bigl(0,\Omega_r\bigr)$, then
\begin{equation}
\sqrt{n}\bigl(\widehat{U}_Q-U_Q(\theta_0)\bigr)
\Rightarrow \min_{j\in J_0} g_{U,j}^{\top} Z,\qquad
\sqrt{n}\bigl(\widehat{L}_Q-L_Q(\theta_0)\bigr)
\Rightarrow \max_{k\in K_0} g_{L,k}^{\top} Z,
\end{equation}
where the generalized gradients at $\theta_0$ are
\begin{equation*}
g_{U,j}=\frac{a_j}{h_Q(\theta_0)} - \frac{\min_{\ell}\bigl(a_{\ell}^{\top}\theta_0+b_{\ell}\bigr)}{h_Q(\theta_0)^2}\, \nabla h_Q(\theta_0), \qquad  j\in J_0.
\end{equation*}

\begin{equation*}
g_{L,k}
=\frac{c_k}{h_Q(\theta_0)}- \frac{\max_{\ell}\bigl(c_{\ell}^{\top}\theta_0+d_{\ell}\bigr)}{h_Q(\theta_0)^2}\,\nabla h_Q(\theta_0),   \qquad  k\in K_0.
\end{equation*}

\qed
\end{theorem}

\begin{remark}
If $|J_0|=1$ and $|K_0|=1$, then the maps $U_Q(\theta)$ and $L_Q(\theta)$ are locally differentiable at $\theta_0$ and the asymptotic results reduce to Theorem \ref{thm1}.
If $|J_0|>1$ or $|K_0|>1$, then in general, the limiting distribution in Theorem \ref{thm2} is not Gaussian. Instead, it is the distribution of a minimum or maximum of finitely many correlated Gaussian linear forms evaluated at the Gaussian limit $Z\sim N(0,\Omega_r)$. 
\end{remark}

\begin{proof}
    
We prove the result for the upper bound $U_Q$ only and the proof for lower bound is similar.
For simplicity, denote $m(\theta)=\min_j(a_j^\top\theta+b_j)$ and $r(\theta)=1/h_Q(\theta)$ so that $U(\theta)=r(\theta)m(\theta)$.
Let $J_0=\arg\min_j(a_j^\top\theta_0+b_j)$.
Take any sequence $t_n\downarrow 0$ and any $h_n\to h$, and define $\theta_n=\theta_0+t_n h_n$.
Let $\delta_j=(a_j^\top\theta_0+b_j)-m(\theta_0)\ge 0$. Then,
\begin{equation*}
\frac{m(\theta_n)-m(\theta_0)}{t_n}
=\min_j\left(a_j^\top h_n+\frac{\delta_j}{t_n}\right).
\end{equation*}
If $j \notin J_0$, then $\delta_j > 0$. Since $t_n \downarrow 0$, we have $\delta_j/t_n \to \infty$, and thus
$a_j^{\top} h_n + \delta_j/t_n \to \infty$. Therefore, for all sufficiently large $n$, the minimum is attained within $J_0$, i.e.
\begin{equation*}
\min_{j}\left(a_j^{\top} h_n + \frac{\delta_j}{t_n}\right)
= \min_{j \in J_0}\left(a_j^{\top} h_n + \frac{\delta_j}{t_n}\right).
\end{equation*}
Thus,
\begin{equation*}
\lim_{n\to\infty}\frac{m(\theta_n)-m(\theta_0)}{t_n}
=\min_{j\in J_0} a_j^\top h
=:m'_{\theta_0}(h).
\end{equation*}
By assumption, $h_Q$ is differentiable at $\theta_0$ with $h_Q(\theta_0)>0$, so the map $r(\theta)=1/h_Q(\theta)$ is differentiable at $\theta_0$ and
\begin{equation*}
r'_{\theta_0}(h)=-\frac{\nabla h_Q(\theta_0)^\top h}{h_Q(\theta_0)^2}.
\end{equation*}
Using $U(\theta)=r(\theta)m(\theta)$ and the decomposition
\begin{equation*}
\begin{aligned}
\frac{U(\theta_n)-U(\theta_0)}{t_n}
=&r(\theta_n)\frac{m(\theta_n)-m(\theta_0)}{t_n}
+m(\theta_0)\frac{r(\theta_n)-r(\theta_0)}{t_n}\\
& +\frac{\bigl(r(\theta_n)-r(\theta_0)\bigr)\bigl(m(\theta_n)-m(\theta_0)\bigr)}{t_n},
\end{aligned}
\end{equation*}
the last term is $o(1)$ because $r(\theta_n)-r(\theta_0)=O(t_n)$ and $m(\theta_n)-m(\theta_0)=O(t_n)$.
Thus, $U$ is Hadamard directionally differentiable at $\theta_0$ with
\begin{equation}
U'_{\theta_0}(h)=r(\theta_0)\,m'_{\theta_0}(h)+m(\theta_0)\,r'_{\theta_0}(h)
=\frac{1}{h_Q(\theta_0)}\min_{j\in J_0}a_j^\top h
-\frac{m(\theta_0)}{h_Q(\theta_0)^2}\nabla h_Q(\theta_0)^\top h.
\end{equation}
In summary,
\begin{equation}
U'_{\theta_0}(h)=\min_{j\in J_0} g_{U,j}^\top h,
\qquad
g_{U,j}=\frac{a_j}{h_Q(\theta_0)}-\frac{m(\theta_0)}{h_Q(\theta_0)^2}\nabla h_Q(\theta_0),
\quad j\in J_0.
\end{equation}
By the directional delta method (\cite{dumbgen1993nondifferentiable}),
\begin{equation}
\sqrt{n}\bigl(\hat{U}-U(\theta_0)\bigr)\Rightarrow U'_{\theta_0}(Z)
=\min_{j\in J_0} g_{U,j}^\top Z.
\end{equation}
\end{proof}

Theorem \ref{thm2} gives the asymptotic result of $\hat U_Q,\, \hat L_Q$, but it is not directly usable in practice because it depends on the unknown active set $J_0$ and $K_0$. A natural idea is to use the standard bootstrap, but this fails due to the discontinuity of the directional derivative map with respect to $\theta_0$. To address this, we apply the numerical delta method of \cite{fang2019inference}, which consistently approximates the limit law by numerically evaluating directional derivatives via finite differences without needing to identify $J_0$ or $K_0$.

The corresponding confidence intervals are stated in the following corollary:

\begin{corollary}\label{coro2}
Assume the conditions of Theorem \ref{thm2}. Let $\hat{\Omega}_r$ be a consistent estimator of $\Omega_r$, and let $\epsilon_n$ be a sequence satisfying $\epsilon_n \to 0$ and $\sqrt{n}\,\epsilon_n \to \infty$.
Define the numerical directional derivative statistics for $B$ simulated draws $Z^{(b)} \overset{iid}{\sim} N(0,\hat{\Omega}_r)$, $b=1,\ldots,B$:
\begin{equation*}
T_U^{(b)}=\frac{U_Q\!\left(\hat{\theta}+\epsilon_n Z^{(b)}\right)-U_Q(\hat{\theta})}{\epsilon_n}, \qquad T_L^{(b)}=\frac{L_Q\!\left(\hat{\theta}+\epsilon_nZ^{(b)}\right)-L_Q(\hat{\theta})}{\epsilon_n}.
\end{equation*}

Let $q_U(\tau)$ and $q_L(\tau)$ denote the empirical $\tau$-quantiles of $\{T_U^{(b)}\}$ and $\{T_L^{(b)}\}$ respectively. Then, asymptotic $(1-\alpha)$ confidence intervals for $U_Q(\theta_0)$ and $L_Q(\theta_0)$ are
\begin{equation}
CI_U=\left[\hat{U}_Q-\frac{q_U(1-\alpha/2)}{\sqrt{n}},
\ \hat{U}_Q-\frac{q_U(\alpha/2)}{\sqrt{n}}\right].
\end{equation}
\begin{equation}
CI_L=\left[\hat{L}_Q-\frac{q_L(1-\alpha/2)}{\sqrt{n}},
\ \hat{L}_Q-\frac{q_L(\alpha/2)}{\sqrt{n}}\right].
\end{equation}
\qed
\end{corollary}

\begin{remark}
The representation in Lemma~1 is not specific to PoCs. The same form holds for any causal quantity
whose sharp bounds can be written as in the form of equation \eqref{equ:affineform} with $h(\theta)>0$ and finite $J,K$. All asymptotic results in Theorems \ref{thm1} and \ref{thm2} continue to hold
under the same setup and assumptions.    
\end{remark}

\section{Simulation}
In this section we run experiments to show that the proposed number of samples provided by Theorems and Corollary in last section are adequate to obtain PoCs within the desired margin errors. We study the bounds for PNS, which is one of the most frequently used PoCs in causal studies. For comparison, we use the same simulation setup as \cite{li2022probabilities}, which studies the sharp PNS bound given in Section 2. It is trivial to verify that this sharp bound has the form shown in Lemma $1$.

\subsection{Experiment: estimating PNS bound}
Following \cite{li2022probabilities}, we consider two parameter settings that share the same causal graph in Figure \ref{fig:scm} but differ in their coefficients. Here $X$ is a binary treatment with $x=1$, $x'=0$, $Y$ is a binary outcome with $y=1$, $y'=0$, and $Z=(Z_1,\ldots,Z_{20})$ is a set of 20 mutually independent binary covariates. Detailed coefficients are provided in Supplementary materials.

\begin{figure}[!htbp]
    \centering
    \includegraphics[width=0.3\linewidth]{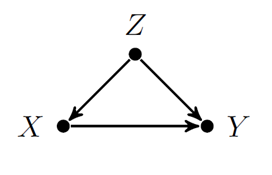}
    \caption{The Causal model for PNS bound experiment. X and Y are binary, and Z is a set of 20 independent binary confounders.}
    \label{fig:scm}
\end{figure}

Let $U_{Z_i}$, $U_X$, and $U_Y$ be independent Bernoulli exogenous variables. The endogenous covariates are defined as follows:
\begin{equation*}
Z_i = U_{Z_i}, \qquad i=1,\ldots,20,
\end{equation*}
and
\begin{equation*}
M_X=\sum_{i=1}^{20} a_i Z_i, \qquad M_Y=\sum_{i=1}^{20} b_i Z_i,
\end{equation*}
where $\{a_i\}$ and $\{b_i\}$ are drawn from Uniform$(-1,1).$
Here are the structural equations and we generate $(X,Y)$ following these equations with a constant $C \sim \text{Uniform}(-1,1)$:
\begin{equation*}
\begin{aligned}
    X=f_X(M_X,U_X)= &
\begin{cases}
1, \quad \text{if } M_X+U_X>0.5,\\
0, \quad  \text{otherwise},
\end{cases}\\
Y=f_Y(X,M_Y,U_Y)= &
\begin{cases}
1,  \quad 0<CX+M_Y+U_Y<1, \,\text{or } 1<CX+M_Y+U_Y<2, \\
0, \quad \text{otherwise}.
\end{cases}
\end{aligned}
\end{equation*}

Since all exogenous variables are binary and independent, the population experimental quantities (\eg, $P(y_x)$ and $P(y_{x'})$) and the population observational quantities (\eg, $P(x,y')$ and $P(x',y)$) can be computed exactly by summing over all configurations of $(U_X,U_Y,U_{Z_1},\ldots,U_{Z_{20}})$. For instance, we can compute
\begin{equation*}
P(y_x)=
\sum_{u_X,u_Y,\mathbf{u}_Z}
P(u_X)\,P(u_Y)\,
\prod_{i=1}^{20} P(u_{Z_i}) 
f_Y\!\bigl(1,\,M_Y(\mathbf{u}_Z),\,u_Y\bigr),
\end{equation*}
and
\begin{equation*}
P(x,y)=
\sum_{u_X,u_Y,\mathbf{u}_Z}
P(u_X)\,P(u_Y)\,
\prod_{i=1}^{20} \Bigl[P(u_{Z_i}) \cdot
f_X\!\bigl(M_X(\mathbf{u}_Z),u_X\bigr)\,\cdot
f_Y\!\Bigl(f_X,\,M_Y(\mathbf{u}_Z),\,u_Y\Bigr) \Bigr],
\end{equation*}
where $\mathbf{u}_Z=(u_{Z_1},\ldots,u_{Z_{20}})$ and $M_X(\mathbf{u}_Z), M_Y(\mathbf{u}_Z)$ denote the corresponding realizations of $M_X$ and $M_Y$. 

We then generate finite samples under two rules. For the experimental sample of size $m$, we draw $(U_X,U_Y,U_{Z_1},\ldots,U_{Z_{20}})$ from their Bernoulli distributions, assign $X\sim \mathrm{Bernoulli}(0.5)$ independently of the exogenous variables, and set $Y=f_Y(X,M_Y,U_Y)$. This will give i.i.d.\ experimental pairs $(X,Y)$. Let $m_{ab}$ be the number of experimental samples with $(X=a,Y=b)$ and $m_{a\cdot}=m_{a0}+m_{a1}$. We estimate
\begin{equation*}
\widehat{P}(y_x)=\widehat{P}(Y=1\mid X=1)=\frac{m_{11}}{m_{1\cdot}}, \qquad \widehat{P}(y_{x'})=\widehat{P}(Y=1\mid X=0)=\frac{m_{01}}{m_{0\cdot}}.
\end{equation*}
For the observational sample of size $n$, we again draw the exogenous variables, set $X=f_X(M_X,U_X)$ and $Y=f_Y(X,M_Y,U_Y)$, and estimate the required joint probabilities by empirical frequencies. Let $n_{ab}$ be the number of observational samples with $(X=a,Y=b)$ and $n=\sum_{a\in\{0,1\}}\sum_{b\in\{0,1\}} n_{ab}$. For example,
\begin{equation*}
\widehat{P}(x,y)=\frac{n_{11}}{n},
\quad
\widehat{P}(x,y')=\frac{n_{10}}{n},
\quad
\widehat{P}(x',y)=\frac{n_{01}}{n}.
\end{equation*}
These estimated probabilities are then plugged into the bound formulas to get finite-sample bound estimates.

Assume the ratio of experimental data to observational data is $m/n \rightarrow r=1$. Then, Corollary \ref{coro1} gives adequate experimental size $m\ge 1921$, while result provided in \cite{li2022probabilities} gives $m \ge 6147$. Our method requires only about $31\%$ of the sample size suggested by the existing result. The detailed computation is in Appendix.

\begin{figure}[!htbp]
\centering
\includegraphics[width=0.33\linewidth]{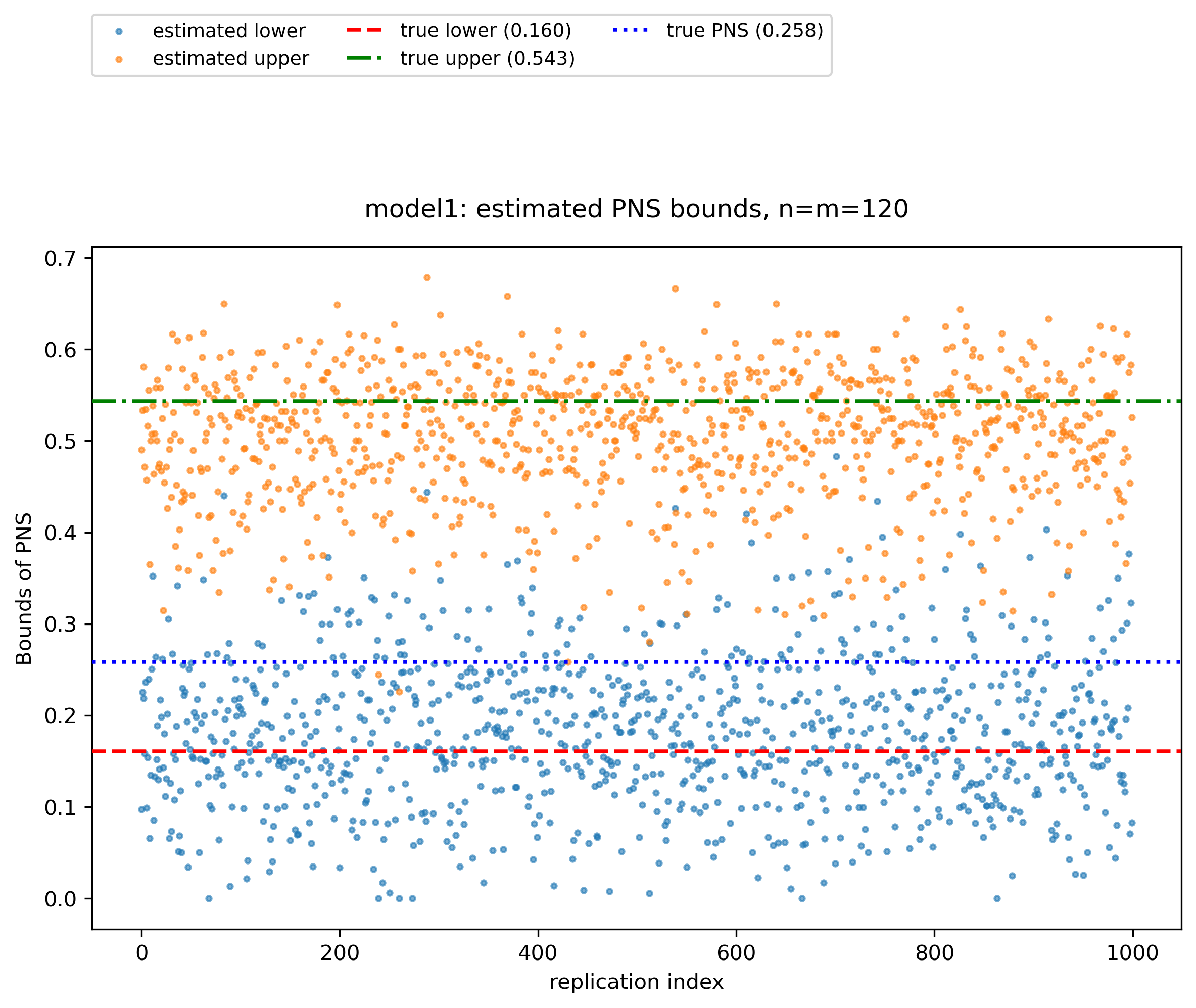}\hfill
\includegraphics[width=0.33\linewidth]{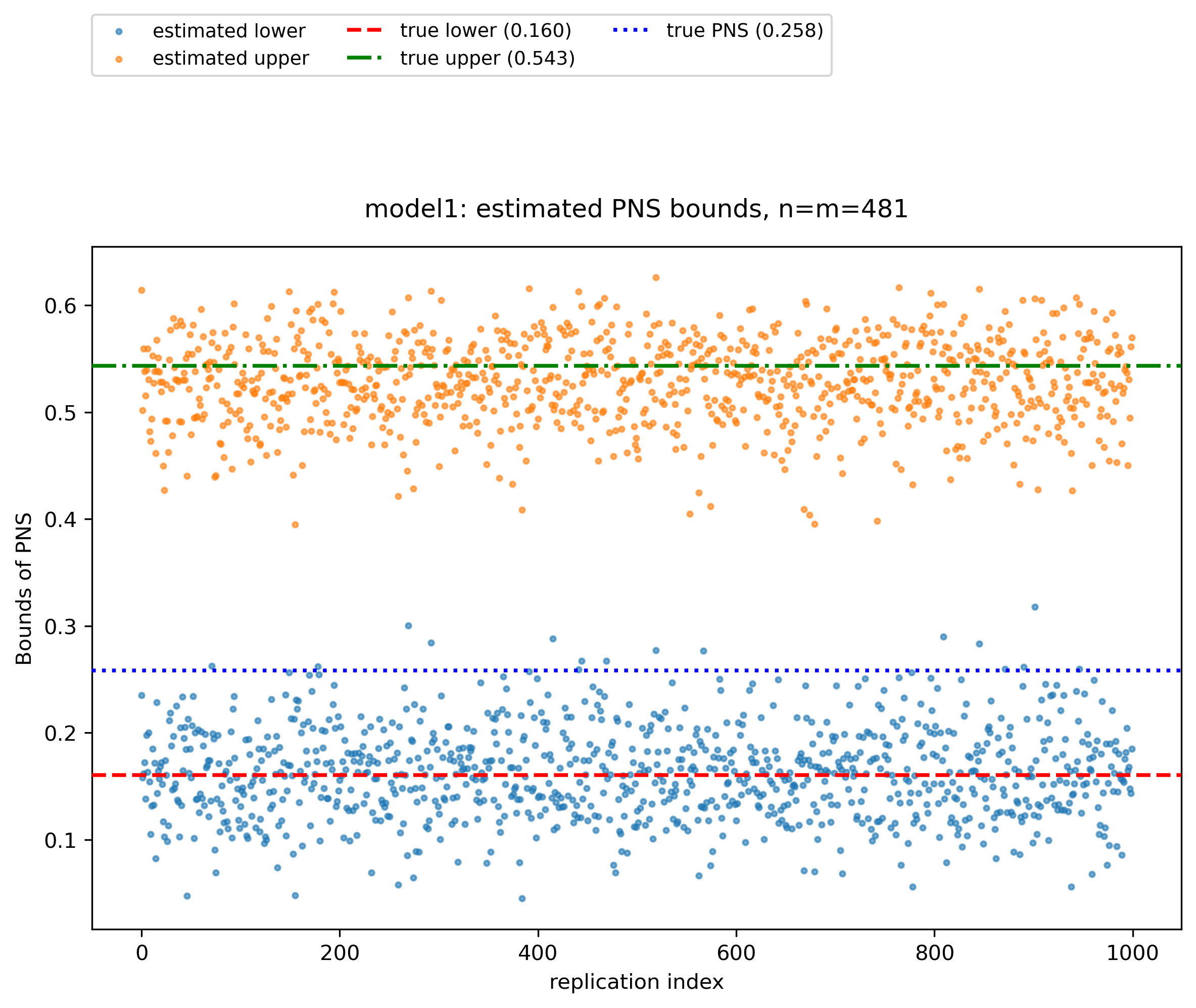}\hfill
\includegraphics[width=0.33\linewidth]{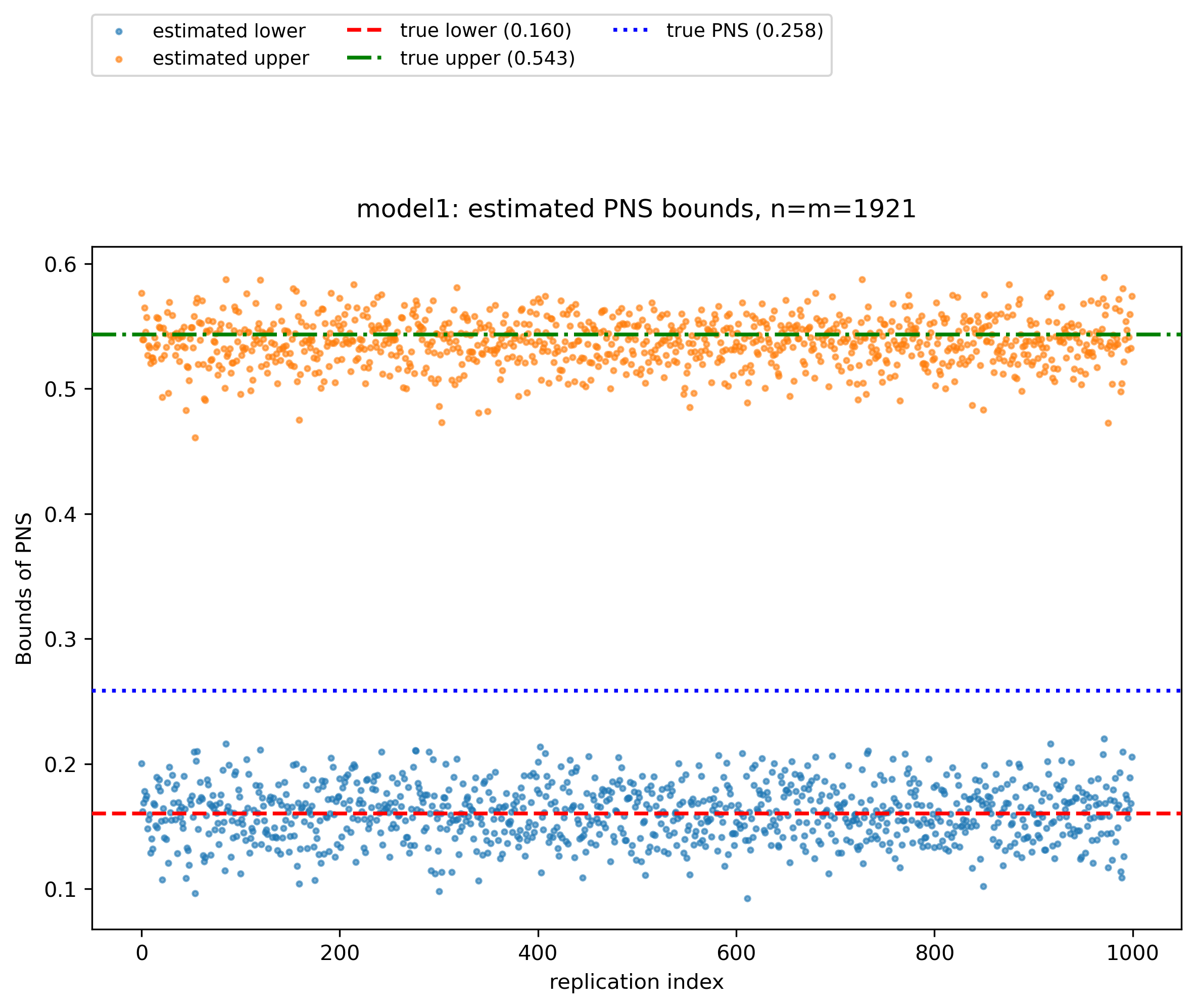}
\caption{Scatter plots under different sample sizes for Model 1. Left to right: $n=120, 481, 1921$.}
\label{fig:model1_scatter}
\end{figure}

For both Model 1 and 2, we run 1000 replications, and in each replication, we generate experimental and observational samples with sizes $n \in \{120, 481, 1921\}$. The results are shown in Figures \ref{fig:model1_scatter} and \ref{fig:model2_scatter}. When $n=120$, both bounds are highly noisy with large dispersion across replications. Increasing the sample size to $n=481$ reduces the noise, but the variance remains relatively large. When 
$n=1921$, the adequate size we computed using Corollary $1$, we can see that the estimates become very stable, with both the upper and lower bound concentrates tightly around the true bounds, which shows a good accuracy and stability.

\begin{figure}[!htbp]
\centering
\includegraphics[width=0.33\linewidth]{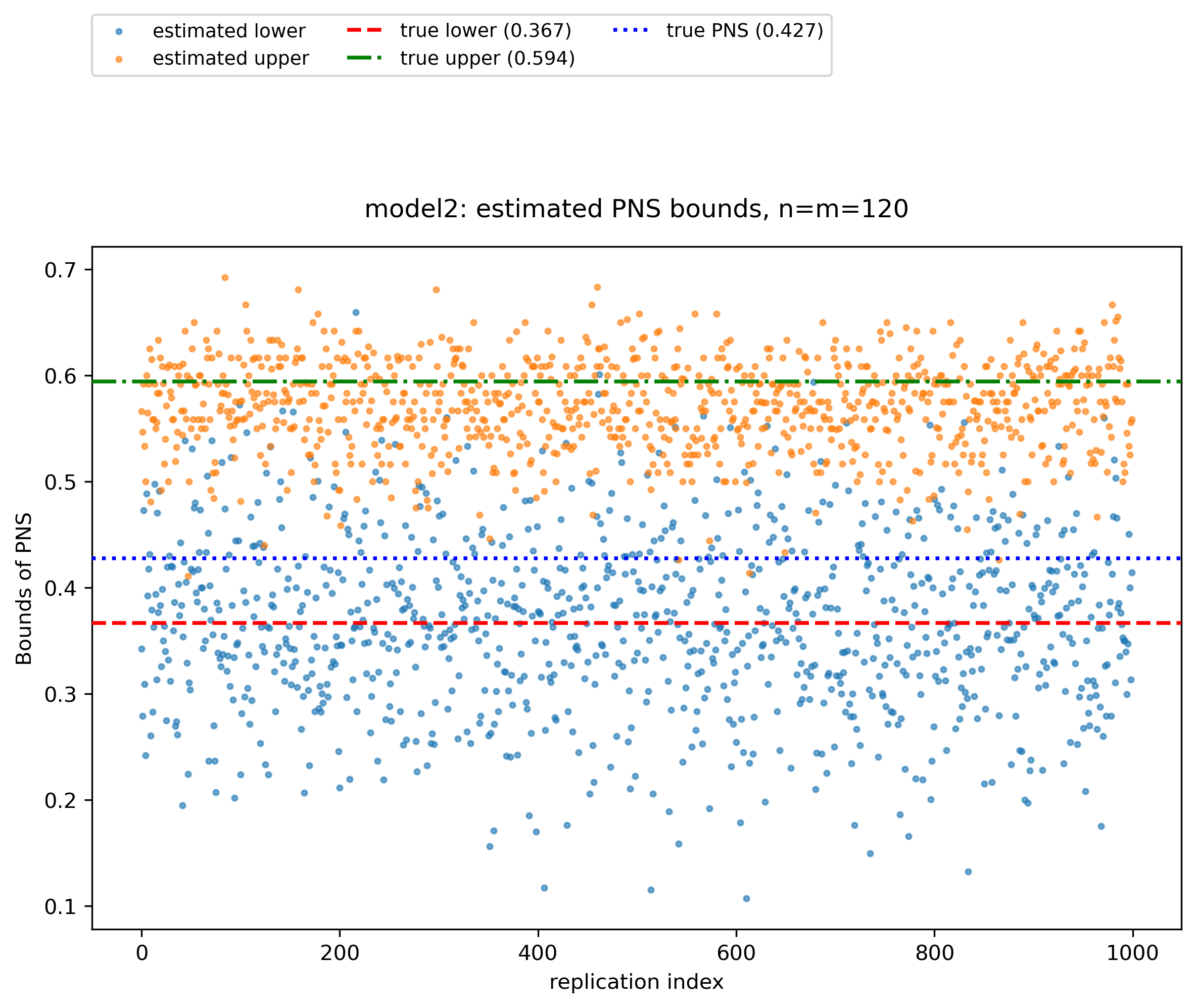}\hfill
\includegraphics[width=0.33\linewidth]{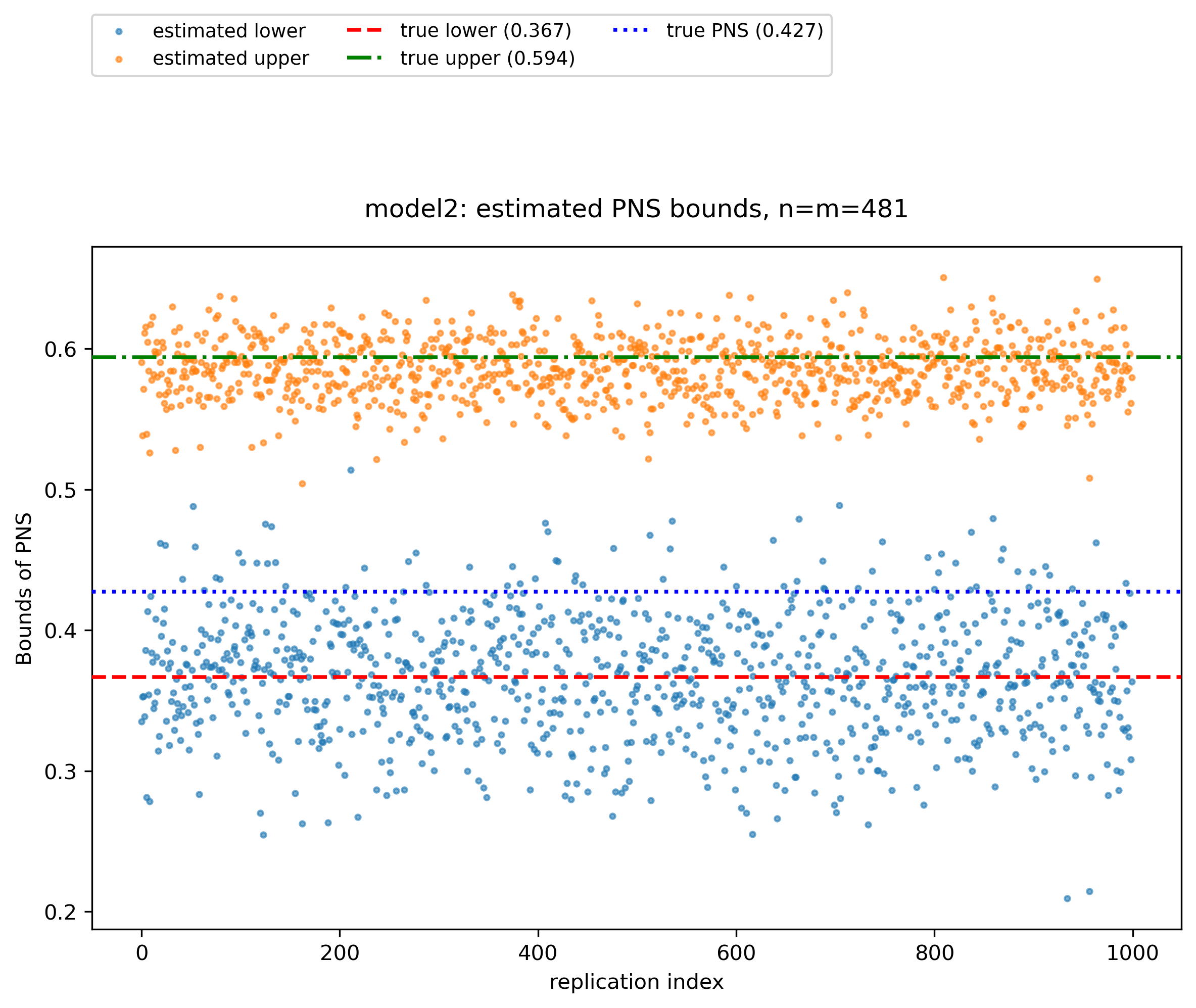}\hfill
\includegraphics[width=0.33\linewidth]{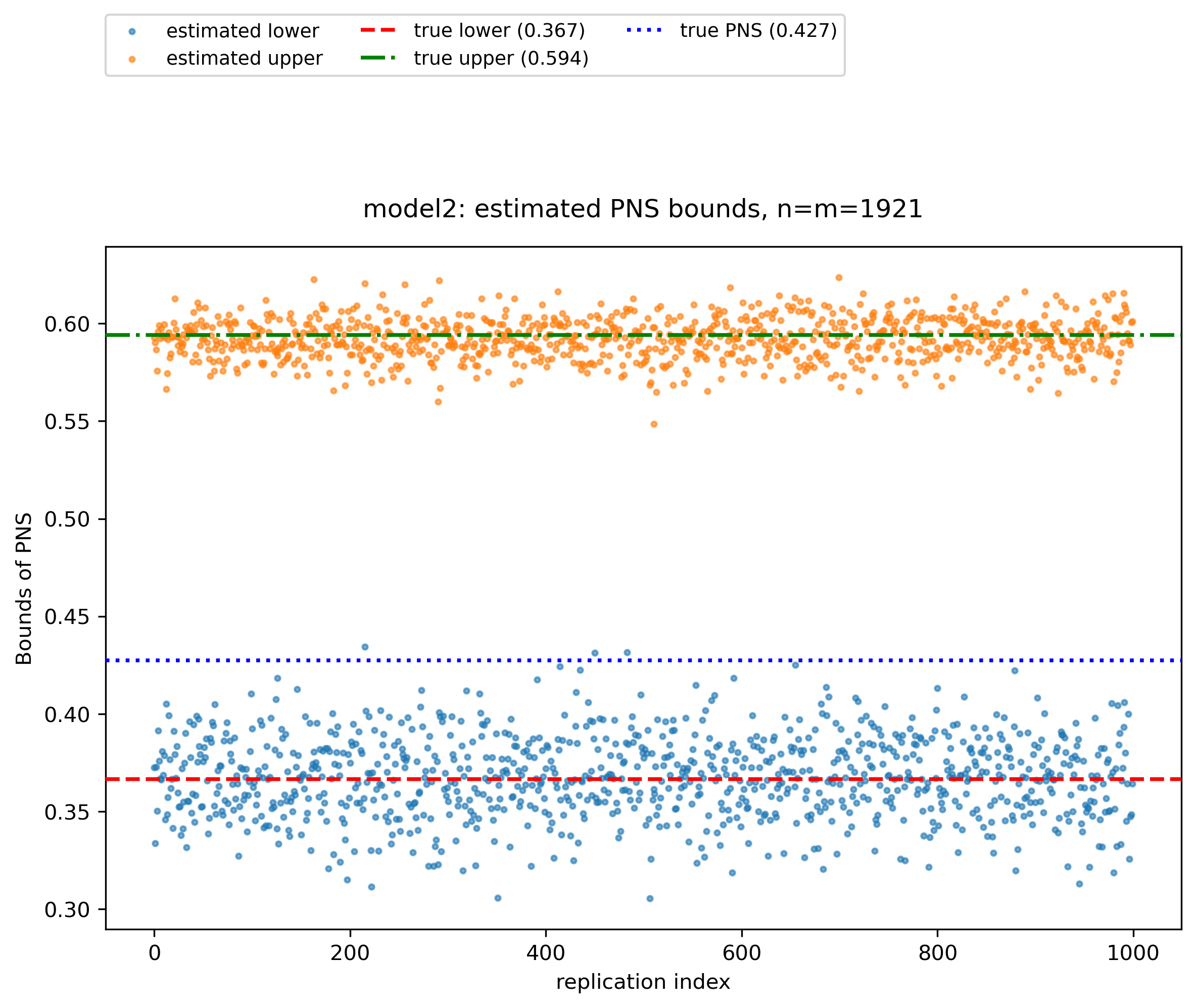}
\caption{Scatter plots under different sample sizes for Model 2. Top to bottom: $n=120, 481, 1921$.}
\label{fig:model2_scatter}
\end{figure}

We also examine the relationship between sample size $n$ and the average estimation error, defined as $(\widehat{\mathrm{bound}}-\mathrm{true\ bound})$ averaged over 1000 replications for each model. See Figure \ref{fig:pnsm1}. Using the target error level 
$0.05 $(red horizontal line), both Model 1 and Model 2 reach the desired accuracy with sample sizes below $300$ for both the upper and lower bounds (blue/orange curves). The average error drops quickly as $n$ increases up to around $500$, which suggests that $n=500$ can already work well in practice. This also confirms that the theoretical sample size $n=1921$ is more than sufficient and it represents the adequate size that considers the worst case.

\begin{figure}[!htbp]
    \centering
    \includegraphics[width=0.49\linewidth]{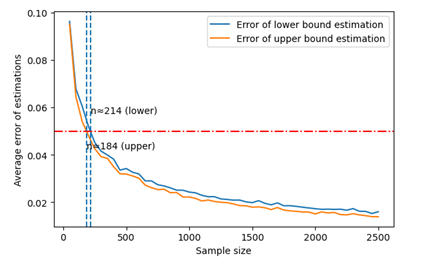} \hfill
    \includegraphics[width=0.49\linewidth]{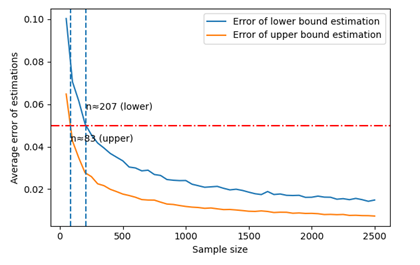}
    \caption{Average error of estimations VS sample size. Left: Model 1, Right:Model 2}
    \label{fig:pnsm1}
\end{figure}

We further investigate the relationship between the average estimation error and the sample size $n$ beyond Models 1 and 2. In order to approximate a more general setting, we repeatedly resample the SCM coefficients and rerun the same simulation design as in Model $1$ for $20 $ independent draws. For each $n$, we then average the estimation error over both the 1000 Monte Carlo replications, and report the resulting curve in Figure \ref{fig:20 replicates}. It shows that both the upper and lower bound estimation errors decrease rapidly as $n$ increases, and in this aggregated experiment the errors converge even faster than in Models 1 and 2. The target error level $0.05$ is reached around $100$, smaller than $500$, which is consistent with our earlier observations. The theoretical sample size $n = 1921$ remains clearly sufficient with very small errors for both bounds. In summary, these results suggest that the sample-size selection rule using our Theorems is not tied to a single parameter setting and remains stable and accurate across a range of SCM with different coefficients.

\begin{figure}[!htbp]
    \centering
    \includegraphics[width=0.6\linewidth]{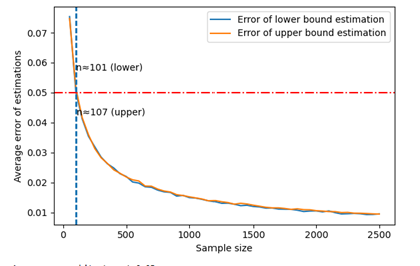}
    \caption{Average error with sample size for 20 replicates}
    \label{fig:20 replicates}
\end{figure}

\section{Discussion}

The bounds of PoCs are piecewise-defined, and points with multiple optimizers lie on the boundaries between pieces. Away from these boundaries, the optimizer is unique and thus Theorem~$1$ applies. Since the measure of boundary points are almost zero, we can think Theorem $1$ holds for most of time. This is supported by our simulations, where the target error is often achieved with far fewer samples than the theoretical results (around $1/4$ to $1/3$), suggesting that the active piece is quite stable in a neighborhood of $\theta_0$.

For the non-unique optimizer case in Theorem~$2$, one may consider a smooth approximation (\eg, Softmin) to enforce differentiability, but this introduces approximation bias and requires  hyperparameter tuning. For this reason we adopt a conservative variance-based approach. However, smoothing may still be useful for large-scale repeated evaluations or gradient-based optimization when near-ties are common and numerical stability is a concern.

Most importantly, although we use PoCs (PNS, PN, PS) in binary setup as motivating examples, Theorems~$1$ and~$2$ apply more broadly. We discuss two extensions here. First, our proposed method can be directly extended to multi-level discrete $X$ and $Y$, with the same CLT and (directional) delta-method arguments but at the cost of higher dimension and computational difficulty. Second, the same methods can be applied to any bounded causal quantity whose sharp bounds can be written as a finite minimum or maximum of linear (or ratio-type) functions of $\theta$, including linear combinations of PoCs such as the “benefit function” in \cite{li2019unit}.

\section{Conclusion}
We study the adequate sample size for estimating bounds of PoCs within a desired margin of error. In finite discrete settings, we first show that the bounds can be written as finite minima or maxima of explicit functions of observational and experimental probabilities, including both linear and ratio-type terms. Using this representation, when the bound endpoint is smooth at the true parameter, we use multivariate delta-method variance approximations to obtain asymptotic margins of error; when the endpoint is not smooth due to active-set switches in the finite min–max structure, we use a directional delta method and implement it using numerical methods following \cite{fang2019inference}. Our asymptotic results apply whenever the bound endpoints has this finite min max form, which covers many bounds in related literature and also extends to other bounded causal quantities, including linear combinations of PoCs. We also provide simulation studies showing that the proposed sample sizes are stable and work well in practice, and are less conservative than existing approaches. To our knowledge, this is the first work providing a general and systematic sample size framework for estimating non-identifiable PoC bounds that handles both smooth and non-smooth endpoints.

\clearpage
\bibliographystyle{apalike}
\bibliography{ref}

\clearpage
\appendix
\section{Appendix}
\subsection{Calculation of sample size in Experiment 1, PNS bound}
We will prove the case of PNS upper bound in \eqref{PNSbound} only. First, focus on the term $P(y_x) - P(y_{x'}) + P(x,y') + P(x',y)$. Denote $A= P(y_x), B= P(y_{x'}), C= P(x,y'), D=P(x',y)$, and we can see the $A$ and $B$ come from randomized controlled trial, which can be assumed to be independent with each other, and independent with $C$ and $D$; however, $C$ and $D$ come from a multinomial distribution, 
\begin{equation*}
(n_{x,y}, n_{x',y}=n_D,n_{x,y'}=n_C,n_{x',y'}) \,\sim\, \text{Multinomial} (n; P(x,y), P(x',y),P(x,y'),P(x',y') ), 
\end{equation*}
such that $C$ and $D$ are negatively correlated.  Combine $C$ and $D$ into one event to make $C \cup D \sim \text{Binomial}(n; P(x,y')+P(x',y)=:P_{C \cup D})$. Then we have:
\begin{equation*}
\begin{aligned}
\text{var} (\hat A -\hat B+\hat C+\hat D) = & \text{var}(\hat A) + \text{var}(\hat B) + \text{var} \widehat{(C+D)}\\
= &  \frac{ P(y_x)(1-P(y_x))}{m_A} + \frac{ P(y_{x'}(1-P(y_{x'})}{m_B} +\frac{P_{C \cup D} (1- P_{C \cup D})}{n}.
\end{aligned}
\end{equation*}
By AM-GM inequality, $p(1-p)\le 1/4$ for any $p \in [0,1]$, thus, 
\begin{equation*}
    \text{var} (\hat A -\hat B+\hat C+\hat D) \le \frac{1}{4m_A}+ \frac{1}{4m_B}+ \frac{1}{4n},
\end{equation*}
with constraint $m_A+m_B=m$.  Assume that $m_A=m_B= m/2$, then
\begin{equation} \label{eq5}
\begin{aligned}
    Err(\hat A -\hat B+\hat C+\hat D)= & z_{1-\alpha/2}\sqrt{\text{var} (\hat A -\hat B+\hat C+\hat D) }\\
    \le & \frac{ z_{1-\alpha/2}}{2} \sqrt{\frac{1}{m_A}+ \frac{1}{m_B}+ \frac{1}{n} }\\
    =& z_{1-\alpha/2}\sqrt{\frac{1}{m}+ \frac{1}{4 n} }.
\end{aligned}
\end{equation}
 
Taking the error bound to be $\epsilon =0.05$, then \eqref{eq5} gives that 
\begin{equation}
    m \ge (1+\frac{1}{4r}) ( \frac{ z_{1-\alpha/2}}{\epsilon} )^2 \approx 1537(1+\frac{1}{4r}).
\end{equation}
When we take $m=n$, i.e. $r=1$, \eqref{eq5} gives adequate experimental size $n=m\ge 1921$, which is the sample size in the worst case.

\subsection{Coefficients of Experiments}
\subsubsection{Model 1}
\begin{equation*}
Z_i = U_{Z_i}, \qquad i \in \{1,\ldots,20\}.
\end{equation*}

\begin{equation*}
X = f_X(M_X,U_X)=
\begin{cases}
1, & \text{if } M_X + U_X > 0.5,\\
0, & \text{otherwise.}
\end{cases}
\end{equation*}

\begin{equation*}
Y = f_Y(X,M_Y,U_Y)=
\begin{cases}
1, & \text{if } 0 < CX + M_Y + U_Y < 1 \ \text{or}\ 1 < CX + M_Y + U_Y < 2,\\
0, & \text{otherwise.}
\end{cases}
\end{equation*}

\begin{equation*}
\begin{aligned}
U_{Z_1} &\sim \mathrm{Bernoulli}(0.352913861526), &\quad U_{Z_2} &\sim \mathrm{Bernoulli}(0.460995855543),\\
U_{Z_3} &\sim \mathrm{Bernoulli}(0.331702473392), &\quad U_{Z_4} &\sim \mathrm{Bernoulli}(0.885505026779),\\
U_{Z_5} &\sim \mathrm{Bernoulli}(0.017026872706), &\quad U_{Z_6} &\sim \mathrm{Bernoulli}(0.380772701708),\\
U_{Z_7} &\sim \mathrm{Bernoulli}(0.028092602705), &\quad U_{Z_8} &\sim \mathrm{Bernoulli}(0.220819399962),\\
U_{Z_9} &\sim \mathrm{Bernoulli}(0.617742227477), &\quad U_{Z_{10}} &\sim \mathrm{Bernoulli}(0.981975046713),\\
U_{Z_{11}} &\sim \mathrm{Bernoulli}(0.142042291381), &\quad U_{Z_{12}} &\sim \mathrm{Bernoulli}(0.833602592350),\\
U_{Z_{13}} &\sim \mathrm{Bernoulli}(0.882938907115), &\quad U_{Z_{14}} &\sim \mathrm{Bernoulli}(0.542143191999),\\
U_{Z_{15}} &\sim \mathrm{Bernoulli}(0.085023436884), &\quad U_{Z_{16}} &\sim \mathrm{Bernoulli}(0.645357252864),\\
U_{Z_{17}} &\sim \mathrm{Bernoulli}(0.863787135134), &\quad U_{Z_{18}} &\sim \mathrm{Bernoulli}(0.460539711624),\\
U_{Z_{19}} &\sim \mathrm{Bernoulli}(0.314014079207), &\quad U_{Z_{20}} &\sim \mathrm{Bernoulli}(0.685879388218),\\
U_X &\sim \mathrm{Bernoulli}(0.601680857267), &\quad U_Y &\sim \mathrm{Bernoulli}(0.497668975278).
\end{aligned}
\end{equation*}

\begin{equation*}
C = -0.77953605542.
\end{equation*}

\begin{equation*}
M_X = [\,Z_1\ Z_2\ \cdots\ Z_{20}\,]\ \beta_X, \qquad
\beta_X =
\begin{bmatrix}
 0.259223510143\\
-0.658140989167\\
-0.75025831768\\
 0.162906462426\\
 0.652023463285\\
-0.0892939586541\\
 0.421469107769\\
-0.443129684766\\
 0.802624388789\\
-0.225740978499\\
 0.716621631717\\
 0.0650682260309\\
-0.220690334026\\
 0.156355773665\\
-0.50693672491\\
-0.707060278115\\
 0.418812816935\\
-0.0822118703986\\
 0.769299853833\\
-0.511585391002
\end{bmatrix}.
\end{equation*}

\begin{equation*}
M_Y = [\,Z_1\ Z_2\ \cdots\ Z_{20}\,]\ \beta_Y, \qquad
\beta_Y =
\begin{bmatrix}
-0.792867111918\\
 0.759967136147\\
 0.55437722369\\
 0.503970540409\\
-0.527187144651\\
 0.378619988091\\
 0.269255196301\\
 0.671597043594\\
 0.396010142274\\
 0.325228576643\\
 0.657808327574\\
 0.801655023993\\
 0.0907679484097\\
-0.0713852594543\\
-0.0691046005285\\
-0.222582013343\\
-0.848408031595\\
-0.584285069026\\
-0.324874831799\\
 0.625621583197
\end{bmatrix}.
\end{equation*}

\subsubsection{Model 2}
\begin{equation*}
Z_i = U_{Z_i}, \qquad i \in \{1,\ldots,20\}.
\end{equation*}

\begin{equation*}
X = f_X(M_X,U_X)=
\begin{cases}
1, & \text{if } M_X + U_X > 0.5,\\
0, & \text{otherwise.}
\end{cases}
\end{equation*}

\begin{equation*}
Y = f_Y(X,M_Y,U_Y)=
\begin{cases}
1, & \text{if } 0 < CX + M_Y + U_Y < 1 \ \text{or}\ 1 < CX + M_Y + U_Y < 2,\\
0, & \text{otherwise.}
\end{cases}
\end{equation*}

\begin{equation*}
\begin{aligned}
U_{Z_1} &\sim \mathrm{Bernoulli}(0.524110233482), &\quad U_{Z_2} &\sim \mathrm{Bernoulli}(0.689566064108),\\
U_{Z_3} &\sim \mathrm{Bernoulli}(0.180145428970), &\quad U_{Z_4} &\sim \mathrm{Bernoulli}(0.317153536644),\\
U_{Z_5} &\sim \mathrm{Bernoulli}(0.046268153873), &\quad U_{Z_6} &\sim \mathrm{Bernoulli}(0.340145244411),\\
U_{Z_7} &\sim \mathrm{Bernoulli}(0.100912238566), &\quad U_{Z_8} &\sim \mathrm{Bernoulli}(0.772038172066),\\
U_{Z_9} &\sim \mathrm{Bernoulli}(0.913108434869), &\quad U_{Z_{10}} &\sim \mathrm{Bernoulli}(0.364272290967),\\
U_{Z_{11}} &\sim \mathrm{Bernoulli}(0.063667554704), &\quad U_{Z_{12}} &\sim \mathrm{Bernoulli}(0.454839320009),\\
U_{Z_{13}} &\sim \mathrm{Bernoulli}(0.586687215140), &\quad U_{Z_{14}} &\sim \mathrm{Bernoulli}(0.018824647595),\\
U_{Z_{15}} &\sim \mathrm{Bernoulli}(0.871017316787), &\quad U_{Z_{16}} &\sim \mathrm{Bernoulli}(0.164966968157),\\
U_{Z_{17}} &\sim \mathrm{Bernoulli}(0.578925020078), &\quad U_{Z_{18}} &\sim \mathrm{Bernoulli}(0.983082980658),\\
U_{Z_{19}} &\sim \mathrm{Bernoulli}(0.018033993991), &\quad U_{Z_{20}} &\sim \mathrm{Bernoulli}(0.074629121266),\\
U_X &\sim \mathrm{Bernoulli}(0.29908139311), &\quad U_Y &\sim \mathrm{Bernoulli}(0.9226108109253).
\end{aligned}
\end{equation*}

\begin{equation*}
C = 0.975140894243.
\end{equation*}

\begin{equation*}
M_X = [\,Z_1\ Z_2\ \cdots\ Z_{20}\,]\beta_X, \qquad
\beta_X =
\begin{bmatrix}
\ \ 0.843870221861\\
\ \ 0.178759296447\\
-0.372349746729\\
-0.950904544846\\
-0.439457721339\\
-0.725970103834\\
-0.791203963585\\
-0.843183562918\\
-0.68422616618\\
-0.782051030131\\
-0.434420454146\\
-0.445019418094\\
\ \ 0.751698021555\\
-0.185984172192\\
\ \ 0.191948271392\\
\ \ 0.401334543567\\
\ \ 0.331387702568\\
\ \ 0.522595634402\\
-0.928734581669\\
\ \ 0.203436441511
\end{bmatrix}.
\end{equation*}

\begin{equation*}
M_Y = [\,Z_1\ Z_2\ \cdots\ Z_{20}\,]\beta_Y, \qquad
\beta_Y =
\begin{bmatrix}
-0.453251661832\\
\ \ 0.424563325534\\
\ \ 0.0924810605305\\
\ \ 0.312680246141\\
\ \ 0.7676961338\\
\ \ 0.124337421843\\
-0.435341306455\\
\ \ 0.248957751703\\
-0.161303883519\\
-0.537653062121\\
-0.222087991408\\
\ \ 0.190167775134\\
-0.788147770713\\
-0.593030174012\\
-0.308066297974\\
\ \ 0.218776507777\\
-0.751253645088\\
-0.11151455376\\
\ \ 0.785227235182\\
-0.568046522383
\end{bmatrix}.
\end{equation*}

\end{document}